\documentclass{article}

\usepackage{footnote}
\usepackage[top=3cm,left=3cm,right=3cm,bottom=3cm]{geometry}
\usepackage{complexity}
\usepackage[top=3cm,left=3cm,right=3cm,bottom=3cm]{geometry}
\usepackage[intlimits]{amsmath}
\usepackage{amsfonts,amssymb,amsthm,bbm, mathtools}
\usepackage[colorlinks=true]{hyperref}
\usepackage{verbatim}
\usepackage{enumerate}
\usepackage[vlined,ruled, linesnumbered]{algorithm2e}
\usepackage{xcolor}
\usepackage{todonotes}

\usepackage{thmtools}
\usepackage{thm-restate}

\hypersetup{pdfstartview=FitH}

\newtheorem{theorem}{Theorem}
\newtheorem{lemma}[theorem]{Lemma}
\newtheorem{fact}[theorem]{Fact}

\theoremstyle{definition}
\newtheorem{definition}{Definition}

\newtheorem{claim}[theorem]{Claim}

\newclass{\alt}{alt}
\newclass{\s}{s}
\newclass{\bs}{bs}
\newclass{\fbs}{fbs}
\newclass{\fC}{fC}
\newclass{\Cert}{C}
\newclass{\EC}{EC}
\newclass{\SB}{sb}
\newclass{\fsb}{fsb}
\newclass{\sC}{sC}
\newclass{\qa}{q_{adv}}
\newclass{\disc}{disc}
\newclass{\bias}{bias}
\newclass{\snip}{snip}
\newclass{\Snip}{Snip}
\newclass{\codim}{codim}
\newclass{\I}{I}
\newclass{\Hen}{H}
\newclass{\Div}{D}
\newclass{\nq}{NOQUERY}
\newclass{\lft}{left}

\newcommand{\cA}{\mathcal{A}}
\newcommand{\clB}{\mathcal{B}}
\newcommand{\clC}{\mathcal{C}}
\newcommand{\clE}{\mathcal{E}}

\newcommand{\clP}{\mathcal{P}}

\newcommand{\clN}{\mathcal{N}}
\newcommand{\N}{\mathsf{N}}

\newcommand{\email}[1]{\href{mailto:#1}{#1}}

\newcommand{\cL}{\mathcal{L}}

\allowdisplaybreaks

\title{A Composition Theorem via Conflict Complexity}
\author{Swagato Sanyal \thanks{Division of Mathematical Sciences, Nanyang Technological University, Singapore and Centre for Quantum Technologies, National University of Singapore, Singapore. \email{ssanyal@ntu.edu.sg}}}

\begin{document}
\maketitle
\begin{abstract}
Let $\R(\cdot)$ stand for the bounded-error randomized query complexity. We show that for any relation $f \subseteq \{0,1\}^n \times \mathcal{S}$ and partial Boolean function $g \subseteq \{0,1\}^n \times \{0,1\}$, $\R_{1/3}(f \circ g^n) = \Omega(\R_{4/9}(f) \cdot \sqrt{\R_{1/3}(g)})$. Independently of us, Gavinsky, Lee and Santha \cite{newcomp} proved this result. By an example demonstrated in their work, this bound is optimal. We prove our result by introducing a novel complexity measure called the \emph{conflict complexity} of a partial Boolean function $g$, denoted by $\chi(g)$, which may be of independent interest. We show that $\chi(g) = \Omega(\sqrt{\R(g)})$ and $\R(f \circ g^n) = \Omega(\R(f) \cdot \chi(g))$.
\end{abstract}
\section{Introduction}
\label{intro}
Let $f \subseteq\{0,1\}^n \times\mathcal{S}$ be a relation and $g \subseteq \{0,1\}^m \times \{0,1\}$ be a partial Boolean function.  In this work, we bound the bounded-error randomized query complexity of the composed relation $f \circ g^n$ from below in terms of the bounded-error query complexitites of $f$ and $g$. Our main theorem is as follows.
\begin{restatable}[Main Theorem]{thm}{main}
\label{main}
For any relation $f \subseteq \{0,1\}^n \times \mathcal{S}$ and partial Boolean function $g \subseteq \{0,1\}^n \times \{0,1\}$,
\[\R_{1/3}(f \circ g^n) = \Omega\left(\R_{4/9}(f) \cdot \sqrt{\R_{1/3}(g)}\right).\]
\end{restatable}
Prior to this work, Anshu et. al. \cite{fstcomp} proved that $\R_{1/3}(f\circ g^n) = \Omega(\R_{4/9}(f)\cdot\R_{1/2-1/n^4}(g))$. Although in the statement of their result $g$ is stated to be a Boolean function, their result holds even when $g$ is a partial Boolean function.

In the special case of $g$ being a total Boolean function, Ben-David and Kothari \cite{DBLP:conf/icalp/Ben-DavidK16} showed that $\R(f\circ g^n) = \Omega\left(\R(f)\cdot \sqrt{\frac{\R(g)}{\log \R(g)}}\right)$.

Gavinsky, Lee and Santha \cite{newcomp} independently proved Theorem~\ref{main} (possibly with different values for the error parameters). They also prove this bound to be tight by exhibiting an example that matches this bound. We believe that our proof is sufficiently different and significantly shorter and simpler than theirs. We draw on and refine the ideas developed in the works of Anshu et. al. and Ben-David and Kothari to prove our result.

We define a novel measure of complexity of a partial Boolean function $g$ that we refer to as the \emph{conflict complexity} of $g$, denoted by $\chi(g)$ (see Section~\ref{cc} for a definition). This quantity is inspired by the \emph{Sabotage complexity} introduced by ben-David and Kothari. However, the two measures also have important differences. For example, we could show that for any partial function $g$, $\chi(g)$ and $\R(g)$ are related as follows.
\begin{restatable}{thm}{maina}
\label{maina}
For any partial Boolean function $g \subseteq \{0,1\}^n \times \{0,1\}$,
\[\chi(g)=\Omega\left(\sqrt{\R_{1/3}(g)}\right).\]
\end{restatable}
See Section~\ref{cc} for a  proof of Throrem~\ref{maina}. Sabotage complexity is known to be similarly related to the bounded-error randomized query complexity (up to a logarithmic factor) when $g$ is a total Boolean function. For partial Boolean functions, unbounded separation is possible between sabotage complexity and $\R(\cdot)$.

We next prove the following composition theorem.
\begin{restatable}{thm}{mainb}
\label{mainb}
Let $\mathcal{S}$ be an arbitrary set, $f \subseteq \{0,1\}^n \times \mathcal{S}$ be a relation and $g \subseteq \{0,1\}^m \times \{0,1\}$ be a partial Boolean function. Then,
\[\R_{1/3}(f \circ g^n)=\Omega(\R_{4/9}(f) \cdot \chi(g)).\]
\end{restatable}
To prove Theorem~\ref{mainb} we draw on the techniques developed by Anshu et. al. and ben-David and Kothari. See Section~\ref{comp} for a proof of Theorem~\ref{mainb}. Theorem~\ref{main} follows from Theorems~\ref{maina} and~\ref{mainb}.
\section{Preliminaries}
\label{prelims}
A partial Boolean function $g$ is a relation in $\{0,1\}^m \times \{0,1\}$. For $b \in \{0,1\}$, $g^{-1}(b)$ is defined to tbe the set of strings $x$ in $\{0,1\}^n$ for which $(x,b) \in g$ and $(x,\overline{b}) \notin g$. $g^{-1}(0) \cup g^{-1}(0)$ is referred to as the set of valid inputs to $g$. We assume that for all strings $y \notin g^{-1}(0) \cup g^{-1}(1)$, both $(y,0)$ and $(y,1)$ are in $g$. For a string $x \in g^{-1}(0) \cup g^{-1}(1)$, $g(x)$ refers to the unique bit $b$ such that $(x,b) \in g$. All the probability distributions $\mu$ over the domain of a partial Boolean function $g$ in this paper are assumed to be supported entirely on $g^{-1}(0) \cup g^{-1}(1)$. Thus $g(x)$ is well-defined for any $x$ in the support of $\mu$.
\begin{definition}[Bounded-error Randomized Query Complexity]
Let $\mathcal{S}$ be any set. Let $h \subseteq \{0,1\}^k \times \mathcal{S}$ be any relation and $\epsilon \in [0,1/2)$. The 2-sided error randomized query complexity $\R_\epsilon(h)$ is the minimum number of queries made in the worst case by a randomized query algorithm $\mathcal A$ (the worst case is over inputs and the internal randomness of $\mathcal{A}$) that on each input $x \in \{0,1\}^k$ satisfies $\Pr[(x,\mathcal A(x)) \in h] \geq 1 - \epsilon$ (where the probability is over the internal randomness of $\mathcal{A}$). 
\end{definition}
\begin{definition}[Distributional Query Complexity]
 Let $h \subseteq \{0,1\}^k \times \mathcal{S}$ be any relation, $\mu$ a distribution on the input space $\{0,1\}^k$ of $h$, and $\epsilon \in [0,1/2)$. The distributional query complexity $\D^\mu_\epsilon(h)$ is the minimum number of queries made in the worst case (over inputs) by a deterministic query algorithm $\mathcal A$ for which $\Pr_{x \sim \mu}[(x,\mathcal A(x)) \in h] \geq 1 - \epsilon$.
\end{definition}

In particular, if $h$ is a function and $\mathcal{A}$ is a randomized or distributional query algorithm computing $h$ with error $\epsilon$, then $\Pr [h(x)=\mathcal{A}(x)] \geq 1-\epsilon$, where the probability is over the respective sources of randomness.

The following theorem is von Neumann's minimax principle stated for decision trees.
\begin{fact}[minimax principle]
\label{minmax}
For any integer $k$, set $\mathcal{S}$, and relation $h \subseteq \{0,1\}^k \times \mathcal{S}$,
\[\R_\epsilon(h)=\max_{\mu}\D_\epsilon^\mu(h).\]
\end{fact}

Let $\mu$ be a probabilty distribution over $\{0,1\}^k$. $x \sim \mu$ implies that $x$ is a random string drawn from $\mu$. Let $C \subseteq \{0,1\}^k$ be arbitrary. Then $\mu \mid C$ is defined tobe the probability distribution obtained by conditioning $\mu$ on the event that the sampled string belongs to $C$, i.e.,
\[\mu \mid C(x)=\left\{  \begin{array}{ll} $0$ & \mbox{if $x \notin C$} \\
\frac{\mu(x)}{\sum_{y \in C} \mu(y)} & \mbox{if $x \in C$}\end{array}   \right.\]

For a partial Boolean function $g:\{0,1\}^m \rightarrow \{0,1\}$, probability distribution $\mu$ and bit $b$, 
\[\mu_b:=\mu \mid g^{-1}(b).\]
Notice that $\mu_0$ and $\mu_1$ are defined with respect to some Boolean function $g$, which will always be clear from the context.

\begin{definition}[Subcube, Co-dimension]
A subset $\clC$ of $\{0,1\}^m$ is called a subcube if there exists a set $S \subseteq \{1, \ldots, m\}$ of indices and an \emph{assignment function} $A:S \rightarrow \{0,1\}$ such that $\clC=\{x \in \{0,1\}^m:\forall i \in S, x_i=A(i)\}$. The co-dimension $\codim(\clC)$ of $\clC$ is defined to be $|S|$. 
\end{definition}
Now we define composition of two relations.
\begin{definition}[Composition of relations]
\label{def:comp}
We now reproduce from the Section~\ref{intro} the definition of composed relations. Let $f \subseteq \{0,1\}^n \times \mathcal{S}$ and $g \subseteq \{0,1\}^m \times \{0,1\}$ be two relations. The composed relation $f \circ g^n \subseteq \left(\{0,1\}^m\right)^n \times \mathcal{S}$ is defined as follows: For $x=(x^{(1)}, \ldots, x^{(n)}) \in \left(\{0,1\}^m\right)^n$ and $s \in \mathcal{S}$, $(x,s) \in f \circ g^n$ if and only if there exists $b=(b^{(1)}, \ldots, b^{(n)}) \in \{0,1\}^n$ such that for each $i=1, \ldots, n$, $(x^{(i)},b^{(i)}) \in g$ and $(b,s) \in f$.
\end{definition}
We will often view a deterministic query algorithm as a binary decision tree. In each vertex $v$ of the tree, an input variable is queried. Depending on the outcome of the query, the computation goes to a child of $v$. The child of $v$ corresponding to outcome $b$ to the query made is denoted by $v_b$.

It is well known that the set of inputs that lead the computation of a decision tree to a certain vertex forms a subcube. We will denote use the same symbol (e.g. $v$) to refer to a vertex as well as the subcube associated with it.

The depth of a vertex $v$ in a tree is the number of vertices on the unique path from the root of the tree to $v$ in the tree. Thus, the depth of the root is $1$.
\begin{definition}
Let $\cA$ be a decision tree on $m$ bits. Let $\eta_0$ and $\eta_1$ be two probability distributions with disjoint supports. Let $v$ be a vertex in $\cA$. Let variable $x_i$ be queried at $v$. Then, \[\Delta^{(v)}:=\left\{\begin{array}{ll}|\Pr_{x \sim \eta_0} [x_i=0]-\Pr_{x \sim \eta_1} [x_i=0]| &  \mbox{if $v \neq \bot$.} \\ 1 &  \mbox{if $v=\bot$.} \end{array}\right.\] 
\end{definition}
Note that $\Delta^{(v)}$ is defined with respect to distributions $\eta_0$ and $\eta_1$. In our application, we will often consider a decision tree $\cA$, a partial Boolean function $g$ and a probability distributions $\mu$ over the inputs. $\Delta^{(v)}$, for a vertex $v$ of $\cA$, will then be assumed to be with respect to the distributions $(\mu_b \mid v)_{b \in \{0,1\}}$.
%Note that conditioned on the event that the process~\ref{A'} reaches a vertex $v$, the probability that it terminates in the same iteration is equal to $\Delta^{(v)}$. Let $b:=g(x)$.
\begin{claim}
\label{mutin}
Let $\cA$ be a decision tree on $m$ bits. Let $g$ be a partial Boolean function. Let $x \sim \{0,1\}^n$ be sampled from a distribution $\mu$. Let $v$ be a vertex in $\cA$. Let variable $x_i$ be queried at $v$. Then,
\[\I_\mu(g(x) : x_i \mid x \in v) = \I_{\mu \mid v}(g(x) : x_i)\geq 32 \left(\Pr_{x \sim \mu \mid v}[g(x)=0] \cdot \Pr_{x \sim \mu \mid v}[g(x)=1] \cdot \Delta^{(v)}\right)^2,\]
where $\Delta^{(v)}$ is with respect to the distributions $(\mu_b \mid v)_{b \in \{0,1\}}$.
\end{claim}
\begin{proof}[Proof of Claim~\ref{mutin}]
Define $b:=g(x)$. Condition on the event $x \in v$. Let $(b \otimes x_i)$ be the distribution over pairs of bits, where the bits are distributed independently according to the distributions of $b$ and $x_i$ respectively. We use the equivalence: $\I(b : x_i)=\Div((b,x_i) || (b \otimes x_i))$. Now, an application of \emph{Pinsker's inequality} implies that 

\begin{align}
\label{fst}\Div((b,x_i) || (b \otimes x_i)) \geq 2 ||(b,x_i)-(b \otimes x_i)||^2_1.
\end{align}
Next, we bound $|(b,x_i)-(b \otimes x_i)||_1$. To this end, we fix bits $z_1, z_2 \in \{0,1\}$, and bound $|\Pr[(b,x_i)=(z_1,z_2)]-\Pr[(b \otimes x_i)=(z_1,z_2)]|$. We have that,
\begin{align}
\label{t1} \Pr[(b,x_i)=(z_1,z_2)]&=\Pr[b=z_1]\Pr[x_i=z_2 \mid b = z_1].
\end{align}
Now,
\begin{align}
\label{t2} \Pr[(b \otimes x_i)=(z_1,z_2)]&=\Pr[b=z_1]\Pr[x_i=z_2] \nonumber \\
&=\Pr[b=z_1](\Pr[b=z_1]\Pr[x_i=z_2 \mid b=z_1]+&\nonumber \\
& \qquad \qquad \qquad \qquad \qquad \Pr[b=\overline{z_1}]\Pr[x_i=z_2 \mid b=\overline{z_1}]).
\end{align}
Taking the absolute difference of~(\ref{t2}) and~(\ref{t1}) we have that,
\begin{align}
&|\Pr[(b,x_i)=(z_1,z_2)]-\Pr[(b \otimes x_i)=(z_1,z_2)]| \nonumber \\
&=\Pr[b=z_1] \cdot \Pr[b=\overline{z_1}] \cdot \Delta^{(v)}=\Pr[b=0] \cdot \Pr[b=1] \cdot \Delta^{(v)}\label{fin}
\end{align}
The Claim follows by adding~(\ref{fin}) over $z_1, z_2$ and using~(\ref{fst}).
\end{proof}

\section{Conflict Complexity}
\label{cc}
In this section, we introduce a randomized process $\clP$ (formally given in Algorithm~\ref{P}). This process is going to play a central role in the proof of our composition theorem (Theorem~\ref{mainb}). Later in the section, we use $\clP$ to define the \emph{conflict complexity} of a partial Boolean function $g$.

Let $n>0$ be any integer and $\clB$ be any deterministic query algorithm that runs on inputs in $(\{0,1\}^m)^n$. $\clB$ can be though of as just a query procedure that queries various input variables, and then terminates without producing any output. Let $x=(x_i^{(j)})_{{i=1, \ldots, n} \atop {j=1, \ldots, m}}$ be a generic input to $\clB$, and $x_i$ stand for $(x_i^{(j)})_{j=1, \ldots, m}$. For a vertex $v$ of $\clB, v^{(i)}$ denotes the subcube in $v$ corresponding to $x_i$, i.e., $v=\times_{i=1}^n v^{(i)}$. Recall from Section~\ref{prelims} that for $b \in \{0,1\}$, $v_b$ stands for the child of $v$ corresponding to the query outcome being $b$. Let $\mu_0$ and $\mu_1$ be any two probability distributions supported on $g^{-1}(0)$ and $g^{-1}(1)$ respectively. Let $z=(z_1, \ldots, z_n) \in \{0,1\}^n$ be arbitrary. Now consider the  probabilistic process $\clP$ given by Algorithm~\ref{P}. Note that $\clP$ can be thought of as a randomized query algorithm on input $z \in \{0,1\}^n$, where a query to $z_i$ corresponds to an assignment of $0$ to $\mathsf{NOQUERY}_i$ in line~\ref{query}. This view of $\clP$ will be adopted in Section~\ref{comp}.

\begin{algorithm}[!h]\label{P}
\DontPrintSemicolon
\caption{ $\clP$ on $\clB, \mu_0, \mu_1, z$}
%\KwIn{$z \in \{0,1\}^n$}
%\KwOut{$f(z)$ or $\perp$ (abort)}
\For{$1 \leq k \leq n$}
{$\nq_k \gets 1$. \;
$\N_k \gets 0$. \;}
$v \gets $Root of $\clB$ \ \ \ \ \ \ \ \ \ \ \ \ \ \  \ \ \ \ \ \ \ // Corresponds to $\{0,1\}^m$ \;
\While{$v$ is not a leaf of $\clB$}
{
Let $\clB$ query $x_i^{(j)}$ at $v$. \;
\If{$\nq_i = 1$}
{Sample a fresh real number $r \sim [0,1]$ uniformly at random. \label{sampler}\;
\If{$r \leq \min_b \Pr_{x_i \sim \mu_b}[x_i^{(j)}=0 \mid x_i \in v_i]$}
{$v \gets v_0$. \;
}
\ElseIf{$r \geq \max_b \Pr_{x_i \sim \mu_b}[x_i^{(j)}=0 \mid x_i \in v^{(i)}]$}
{$v \gets v_1$. \;
 }
\Else
{$\nq_i \gets 0$. \label{query}\;
\If{$r \leq \Pr_{x_i \sim \mu_{z_i}}[x_i^{(j)}=0 \mid x_i \in v^{(i)}]$}
{$v \gets v_0$. \;}
\Else
{$v \gets v_1$. \;}
}
$ \N_i\gets \N_i+1$. \;}
\Else{
Sample $b$ from the distribution $\mu_{z_i}$ conditioned on the event $x_i \in v^{(i)}$. \;
$v \gets v_b$. \;}
}
\end{algorithm}
We now prove an important structural result about $\clP$ which will be used many times in our proofs. Consider the following distribution $\gamma_z$ on $(\{0,1\}^m)^n$: For each $i$, sample $x_i$ independently from $\mu_{z_i}$.

Let $v$ be a vertex of $\clB$. Let $A_\clB(v)$ be the event that process $\clP$ reaches node $v$, and $B_\clB(v)$ be the event that for a random input $x$ sampled from $\gamma_z$, the computation of $\clB$ reaches node $v$.
\begin{claim}
\label{samedistn}
For each vertex $v$ of $\clB$, \[\Pr[A_\clB(v)]=\Pr[B_\clB(v)].\]
\end{claim}
\begin{proof}
We will prove by induction on the depth $t$ of $v$, i.e., the number of vertices on the unique path from the root to $v$ in $\clB$.
\begin{description}
\item[Base case:] $t=1$. $v$ is the root of $\clB$. Thus $\Pr[A_\clB(v)]=\Pr[B_\clB(v)]=1$.
\item[Inductive step:] Assume that $t \geq 2$, and that the statement is true for all vertices at depth at most $t-1$. Since $t \geq 2$, $v$ is not the root of $\clB$. Let $u$ be the ancestor of $v$, and variable $x_i^{(j)}$ be queried at $u$. without loss of generality assume that $v$ is the child of $u$ corresponding to $x_i^{(j)}=0$. We split the proof into the following two cases.
\begin{itemize}
\item {\bf Case 1:} $\Pr_{x_i \sim \mu_{z_i}}[x_i^{(j)}=0 \mid x_i \in u_i] \leq \Pr_{x_i \sim \mu_{\overline{z_i}}}[x_i^{(j)}=0 \mid x_i \in u_i]$.

Condition on $A_\clB(u)$ and $\nq_i=0$. The probability that $\clP$ reaches $v$ is $\Pr_{x_i \sim \mu_{z_i}}[x_i^{(j)}=0 \mid x_i \in u_i]$. Now, condition on $A_\clB(u)$ and $\nq_i=1$. The probability that $\clP$ reaches $v$ is exactly equal to the probability that the real number $r$ sampled at $v$ lies in $[0, \Pr_{x_i \sim \mu_{z_i}}[x_i^{(j)}=0 \mid x_i \in u_i] ]$, which is equal to $\Pr_{x_i \sim \mu_{z_i}}[x_i^{(j)}=0 \mid x_i \in u_i]$. Thus,
\begin{align}
\Pr[A_\clB(v]&=\Pr[A_\clB(u)]. \Pr[A_\clB(v) \mid A_\clB(u)] \nonumber \\
&=\Pr[A_{\clB}(u)] \cdot \Pr_{x_i \sim \mu_{z_i}}[x_i^{(j)}=0 \mid x_i \in u_i]. \label{c1:one}
\end{align}
Now condition on $B_\clB(u)$. The probability that $\clB$ reaches $v$ is exactly equal to the probability that $x_i^{(j)}=0$ when $x$ is sampled according to the distribution $\gamma_z$ conditioned on the event that $x \in u$. Note that in the distribution $\gamma_z$, the $x_k$'s are independently distributed. Thus,
\begin{align}
\Pr[B_\clB(v)]&=\Pr[B_\clB(u)]. \Pr[B_\clB(v) \mid B_\clB(u)] \nonumber \\
&=\Pr[B_{\clB}(u)] \cdot \Pr_{x_i \sim \mu_{z_i}}[x_i^{(j)}=0 \mid x_i \in u_i]. \label{c1:two}
\end{align}
By the inductive hypothesis, $\Pr[A_{\clB}(u)]=\Pr[B_{\clB}(u)]$. The claim follows from~(\ref{c1:one}) and~(\ref{c1:two}).
\item {\bf Case 2:} $\Pr_{x_i \sim \mu_{z_i}}[x_i^{(j)}=0 \mid x_i \in u_i] > \Pr_{x_i \sim \mu_{\overline{z_i}}}[x_i^{(j)}=0 \mid x_i \in u_i]$.
Let $v'$ be the child of $u$ corresponding to $x_i^{(j)}=1$. By an argument similar to Case 1, we have that 
\begin{align}
\Pr[A_\clB(v')]=\Pr[B_\clB(v')]. \label{c2}
\end{align}
Now,
\begin{align}
\Pr[A_\clB(v)] &=\Pr[A_\clB(u)] - \Pr[A_\clB(v')] \nonumber \\
&= \Pr[B_{\clB}(u)] - \Pr[A_\clB(v')]  \mbox{\ \ \ \ \ (By inductive hypothesis)} \nonumber \\
&= \Pr[B_\clB(u)] - \Pr[B_\clB(v')] \mbox{\ \ \ \ \ \ (By (\ref{c2}))} \nonumber \\
&= \Pr[B_\clB(v]. \nonumber
\end{align}
\end{itemize}
\end{description}
\end{proof}
Let $n=1, z \in \{0,1\}$, and $\clB$ be a decision tree that computes $g$. Consider process $\clP$ on $\clB, \mu_0, \mu_1, z$. Note that $\nq_1$ is set to $0$ with probability $1$. To see this observe that as long as $\nq_1=1$, the current subcube $v$ contains strings from the supports of both $\mu_0$ and $\mu_1$, and hence from both $g^{-1}(0)$ and $g^{-1}(1)$. If $\nq_1$ is not set to $0$ for the entire run of $\clP$, then there exist inputs $x \in g^{-1}(0), x' \in g^{-1}(1)$ which belong to the same leaf of $\clB$, contradicting the hypothesis that $\clB$ computes $g$. Let the random variable $\clN$ stand for the value of the variable $\N_1$ after the termination of $\clP$. Note that $\clN$ is equal to the the index of the iteration of the while loop in which $\nq_1$ is set to $0$. The distribution of $\clN$ depends on $\mu_0, \mu_1$ and $\clB$, which in our applications will either be clear from the context, or clearly specified. Note that the distribution of $\clN$ is independent of the value of $z$.

\begin{definition}
The \emph{conflict complexity} of a partial Boolean function $g$ with respect to distributions $\mu_0$ and $\mu_1$ supported on $g^{-1}(0)$ and $g^{-1}(1)$ respectively, and decision tree $\clB$ computing $g$, is defined as:
\[\chi (\mu_0, \mu_1, \clB)=\E[\clN].\footnote{As observed before, the choices of $\mu_0, \mu_1$ and $\clB$ are built into the definition of $\clN$.}\]
The conflict complexity of $g$ is defined as:
\[\chi(g)=\max_{\mu_0, \mu_1} \min_\clB \chi(\mu_0, \mu_1, \clB).\]
Where the maximum is over distributions $\mu_0$ and $\mu_1$ supported on $g^{-1}(0)$ and $g^{-1}(1)$ respectively, and the minimum is over decision trees $\clB$ computing $g$.
\end{definition}
For a pair $(\mu_0, \mu_1)$ of distributions, let $\clB$ be the decision tree computing $g$ such that $\E[\clN]$ is minimized. We call such a decision tree an \emph{optimal} decision tree for $\mu_0, \mu_1$. We conclude this section by making an important observation about the structure of optimal decision trees. Let $v$ be any node of $\clB$. Let $\mu_0':=\mu_0 \mid v$ and $\mu_1':=\mu_1 \mid v$. Let $\clB_v$ denote the subtree of $\clB$ rooted at $v$. We observe that $\clB_v$ is an optimal  tree for $\mu_0'$ and $\mu_1'$; if it is not then we could replace it by an optimal tree for $\mu_0'$ and $\mu_1'$, and for the resultant tree, the expected value of $\clN$ with respect to $\mu_0$ and $\mu_1$ will be smaller than that in $\clB$. This will contradict the optimality of $\clB$. This recursive sub-structure property of optimal trees will be helpful to us.
\section{Conflict Complexity and Randomized Query Complexity}
In this section, we will prove Theorem~\ref{maina} (restated below).
\maina*
\begin{proof}
\label{hibias}
We will bound the distributional query complexity of $g$ for each input distribution $\mu$ with rspect to error $47/95<1/2$, $\D_{47/95}^\mu (g)$, from above by $O(\chi(g)^2)$. Theorem~\ref{maina} will follow from the \emph{minimax principle} (Fact~\ref{minmax}), and the observation that the error can be brought down to $1/3$ by constantly many independent repetitions followed by a selection of the majority of the answers. It is enough to consider distributions $\mu$ supported on valid inputs of $g$. To this end, fix a distribution $\mu$ supported only on $g^{-1}(0) \cup g^{-1}(1)$.

Let $\chi(g)=d$. Let $\mu_b$ be the distribution obtained by conditioning $\mu$ on the event $g(x)=b$. Let $\clB$ be an optimal decision tree for distributions $\mu_0$ and $\mu_1$. Clearly $\E[\clN] \leq \chi(g) = d$.

We first prove some structural results about $\clB$. Let $\clB$ be run on a random input $x$ sampled according to $\mu$. Let $v_t$ be the random vertex at which the $t$-th query is made; If $\clB$ terminates before making $t$ queries, define $v_t:=\bot$. Let $\clE$ be any event which is a collection of possible transcripts of $\clB$, such that $\Pr[\clE] \geq \frac{3}{4}$. Recall from Section~\ref{prelims} that for any vertex $v$ of $\clB$, $\Delta^{(v)}$ is assumed to be with respect to the probability distribution $\mu \mid v$.
\begin{claim}
\label{sumofdelta}
\[\sum_{t=1}^{10d} \E[\Delta^{(v_t)} \mid \clE]  \geq \frac{13}{20}.\]
\end{claim}
\begin{proof}
Let us sample vertices $u_t$ of $\clB$ as follows:
\begin{enumerate}
\item Set $z=\left\{\begin{array}{ll}
1 & \mbox{with probability $\Pr_{x \sim \mu}[g(x)=1]$}, \\
0 & \mbox{with probability $\Pr_{x \sim \mu}[g(x)=0]$}
\end{array}\right.$
\item Run process $\clP$ for $\clB, \mu_0, \mu_1, z$.
\item Let $u_t$ be the vertex $v$ in the beginning of the $t$-th iteration of the \emph{while} loop of Algorithm~\ref{P}. Return $(u_t)_{t=1,\ldots}$. If the simulation stops after $i$ iterations, set $u_t:=\bot$ for all $t > i$.
\end{enumerate}
By Claim~\ref{samedistn}, and since $z$ has the same distribution as that of $g(x)$ where $x$ is sampled from $\mu$, the vertices $u_t$ and $v_t$ have the same distribution. 
In the above sampling process for each $t=1, \ldots, 10d$, let $E_t$ be the event that $\nq_1=1$ in the beginning of the $t$-th iteration of the \emph{while} loop of Algorithm~\ref{P}. Conditioned on $\clE$, the probability that $\nq_1$ is set to $0$ in the $t$-th iteration is $\Pr[E_t \mid \clE] \cdot \E[\Delta^{(u_t)} \mid E_t, \clE]$\footnote{Note that conditioned on $E_t$, $u_t \neq \bot$.}. By union bound we have that,
\begin{align}
\sum_{t=1}^{10d} \E[\Delta^{(v_t)} \mid \clE]&=\sum_{t=1}^{10d} \E[\Delta^{(u_t)} \mid \clE] \nonumber \\
&\geq \sum_{t=1}^{10d}\Pr[E_t \mid \clE] \cdot  \E[\Delta^{(u_t)} \mid E_t, \clE] \nonumber \\
&\geq \Pr\left[\overline{\bigcap_{t=1}^{10d} E_t} \mid \clE \right]  \label{delbound1} \nonumber \\
&\geq \Pr\left[\overline{\bigcap_{t=1}^{10d} E_t}\right] - \Pr[\overline{\clE}].
\end{align}
Now, since $\E[\clN] \leq \chi(g) = d$, we have by \emph{Markov's inequality} that the probability that the process $\clP$, when run for $\clB, \mu_0, \mu_1$ and the random bit $z$ generated as above\footnote{Recall that the distribution of $\clN$ is independent of $z$.}, sets $\nq_1$ to $0$ within first $10d$ iterations of the \emph{while} loop, is at least $9/10$. Thus we have that,
\begin{align}
\Pr\left[\bigcap_{t=1}^{10d} E_t\right]^c \geq \frac{9}{10}. \label{delbound2}
\end{align}
The claim follows from (\ref{delbound1}), (\ref{delbound2}) and the hypothesis $\Pr[\clE] \geq \frac{3}{4}$.
\end{proof}
The next Lemma follows from Claim~\ref{sumofdelta} and the recursive sub-structure property of optimal trees discussed in the last paragraph of Section~\ref{cc}.
\begin{lemma}
\label{sumofdelta2}
Let $i$ be any positive integer. Then,
\[\sum_{t=1}^{10di} \E[\Delta^{(v_t)} \mid \clE]  \geq \frac{13i}{20}.\]
\end{lemma}
Notice that if $\clB$ terminates before making $t$ queries, $v_t=\bot$ and $\Delta^{(v_t)}=1$.
\begin{proof}[Proof of Lemma~\ref{sumofdelta2}]
For $j=1,\ldots, i$, let $w$ be any vertex at depth $10jd+1$. Consider the subtree $\mathsf{T}$ of $\clB$ rooted at $w$. By the recursive sub-structure property of $\clB$, $\mathsf{T}$ is an optimal tree for distributions $\mu_0':=\mu_0 \mid w, \mu_1':=\mu_1 \mid w$. Let $w_t$ be the random vertex at depth $t$ of $\mathsf{T}$, when $\mathsf{T}$ is run on a random input from $\mu \mid w$. By Claim~\ref{sumofdelta}, we have that,
\begin{align}
\sum_{t=1}^{10d} \E[\Delta^{(w_t)} \mid \clE]  \geq \frac{13}{20}. \label{onetree}
\end{align}
In~(\ref{onetree}), $\Delta^{(w_t)}$ is with respect to distributions $\mu'_0 \mid w_t=\mu_0 \mid w_t, \mu'_1 \mid w_t=\mu_1 \mid w_t$. Now, when $w$ is the random vertex $v_{10jd+1}$, $w_t$ is the random vertex $v_{10jd+t}$. Thus from (\ref{onetree}) we have that,
\begin{align}
\sum _{t=10jd+1}^{10(j+1)d} \E[\Delta^{(v_t)} \mid \clE]  \geq \frac{13}{20}. \label{oneslab}
\end{align}
The claim follows by adding (\ref{oneslab}) over $j=0, \ldots, i-1$.
\end{proof}

%%%
\begin{comment}
Let $x =(x^{(i)})_{i=1, \ldots, m} \in \{0,1\}^m$ be a generic input to $\clB'$. Let $v$ be a vertex in $\clB'$, and let $x^{(i)}$ be the variable queried at $v$. Define $\Delta_v:=\left|\Pr_{x \sim \mu_0}[x^{(i)}=0 \mid x \in v]-\Pr_{x \sim \mu_1}[x^{(i)}=0 \mid x \in v]\right|$. Conditioned on the event that at the beginning of an iteration of the \emph{while} loop in Algorithm~\ref{P} the current vertex is $v$ and $\nq_1=1$, the probability that $\nq_1$ is set to $0$ in that iteration is exactly $\Delta_v$.
\end{comment}
%%%

We now finish the proof of Theorem~\ref{maina} by showing that $D^\mu(g) = O(d^2)$. Let $x$ be distributed according to $\mu$, and $\clB$ be run on $x$. Let $\mathsf{BIASED}$ denote the event that in at most $10d^2$ queries, the computation of $\clB$ reaches a vertex $v$ for which $\Pr_{x \sim \mu}[g(x)=0 \mid x \in v]\cdot\Pr_{x \sim \mu}[g(x)=1 \mid x \in v] \leq \frac{1}{9}$. Let $\mathsf{STOP}$ denote the event that $\clB$ terminates after making at most $10d^2$ queries. Let $\clE:=\overline{\mathsf{BIASED} \vee\mathsf{STOP}}$.

Consider the following decision tree $\clB'$: Start simulating $\clB$. Terminate the simulation if one of the following events occurs. The outputs in each case is specified below.
\begin{enumerate}
\item (\emph{Event $\mathsf{STOP}$}) If $\clB$ terminates, terminate and output what $\clB$ outputs. \label{e1}
\item If $10d^2$ queries have been made and the computation is at a vertex $v$, terminate and output $\arg \max_b \Pr[g(x)=b \mid x \in v]$. \label{e2}
\end{enumerate}
By construction, $\clB'$ makes at most $10d^2$ queries in the worst case. We shall show that $\Pr_{x \sim \mu}[\clB'(x)\neq g(x)] \leq \frac{47}{95} < \frac{1}{2}$. This will prove Theorem~\ref{maina}.

We split the proof into the following two cases.
\begin{description}
\item [Case $1$:] $\Pr[\overline{\clE}] \geq \frac{1}{4}$.

First, condition on the event that the computation reaches a vertex $v$ for which $\Pr_{x \sim \mu}[g(x)=0 \mid x \in v]\cdot\Pr_{x \sim \mu}[g(x)=1 \mid x \in v] \leq \frac{1}{9}$ holds. Thus one of $\Pr_{x \sim \mu}[g(x)=0 \mid x \in v]$ and $\Pr_{x \sim \mu}[g(x)=1 \mid x \in v]$ is at most $1/3$. Hence,  $|\Pr_{x \sim \mu}[g(x)=0 \mid x \in v]-\Pr_{x \sim \mu}[g(x)=1 \mid x \in v]| \geq 2/3$. Let $m$ be the random leaf of the subtree of $\clB'$ rooted at $v$ at which the computation ends. The probability that $\clB'$ errs is at most
\begin{align*}
&\E_{x \sim \mu \mid v}\left[\frac{1}{2}-\frac{1}{2}\left|\Pr_{x \sim \mu}[g(x)=0 \mid x \in m]-\Pr_{x \sim \mu}[g(x)=1 \mid x \in m]\right|\right]. \\
& \leq \frac{1}{2}-\frac{1}{2} \left|\E_{x \sim \mu \mid v}\Pr_{x \sim \mu}[g(x)=0 \mid x \in m]-\E_{x \sim \mu \mid v}\Pr_{x \sim \mu}[g(x)=1 \mid x \in m]\right| \\
& \qquad \qquad \qquad \qquad \mbox{\ \ \ \ (By Jensen's inequality)} \\
&=\frac{1}{2}-\frac{1}{2}\left|\Pr_{x \sim \mu}[g(x)=0 \mid x \in v]-\Pr_{x \sim \mu}[g(x)=1 \mid x \in v]\right| \leq \frac{1}{3}.
\end{align*}

Then, condition on the event $\mathsf{STOP}$. The probability that $\clB'$ errs is $0 \leq \frac{1}{3}$.

Thus we have shown that conditioned on $\overline{\clE}$ the probability that $\clB'$ errs is at most $\frac{1}{3}$. Thus the probability that $\clB'$ errs is at most $\frac{1}{4}\cdot \frac{1}{3}+\frac{3}{4}\cdot\frac{1}{2} = \frac{11}{24}<\frac{47}{95}$.
\item[Case $2$:] $\Pr[\overline{\clE}] < \frac{1}{4}$.

By Claim~\ref{sumofdelta2} we have that
\begin{align}
\sum_{t=1}^{10d^2} \E[\Delta^{v^{(t)}} \mid \clE] \geq \frac{13d}{20}. \label{deltabound}
\end{align}
Let $a_i:=(x_i, b_i)$ be the tuple formed by the random input variable $x_i$ queried at the $i$-th step by $\clB'$, and the outcome $b_i$ of the query; if $\clB'$ terminates before $i$-th step, $a_i:=\bot$. Notice that the vertex $v_i$ at which the $i$-th query is made is determined by $(a_1, \ldots, a_{i-1})$ and vice versa. We have,
\begin{align}
&\I(a_1, \ldots, a_{10d^2}:g(x)) \nonumber \\
&= \sum_{i=1}^{10d^2} \I(a_i:g(x) \mid a_1, \ldots, a_{i-1}) \mbox{\ \ \ \ (Chain rule of mutual information)}\nonumber \\
&= \sum_{i=1}^{10d^2} \I(b_i:g(x) \mid v_i) \nonumber \\
& \geq 32 \sum_{i=1}^{10d^2} \E \left[\mathbf{1}_{v_i \neq \bot}\cdot\left[\Pr[g(x)=0 \mid x \in v_i] \cdot \Pr[g(x)=1 \mid x \in v_i] \cdot \Delta^{(v_i)}\right]^2\right] \nonumber \\
&\qquad \qquad \qquad \qquad \mbox{\ \ \ (From Claim~\ref{mutin})} \nonumber \\
&\geq 32 \sum_{i=1}^{10d^2} \Pr[\clE] \cdot \E\left[\left[\Pr[g(x)=0 \mid x \in v_{i-1}] \cdot \Pr[g(x)=1 \mid x \in v_{i-1}] \cdot \Delta^{(v_i)}\right]^2 \mid \clE\right] \nonumber \\
&\qquad \qquad \qquad \qquad \mbox{\ \ \ \ (Conditioned on $\clE, v_i \neq \bot$)} \nonumber \\
&\geq 32 \sum_{i=1}^{10d^2} \frac{3}{4} \cdot \frac{1}{9} \cdot \E[{\Delta^{(v_i)}}^2 \mid \clE] \nonumber \\
&= \frac{8}{3}\sum_{i=1}^{10d^2}   \E[{\Delta^{(v_i)}}^2 \mid \clE] \mbox{\ \ \ \ \ \ (By the assumption $\Pr[\overline{\clE}] \leq \frac{1}{4}$ )} \nonumber \\
&\geq \frac{8}{3} \cdot \frac{1}{10d^2} \left(\sum_{i=1}^{10d^2}  \E[\Delta^{(v_i)} \mid \clE]\right)^2 \mbox{(By Cauchy-Schwarz inequality)} \nonumber \\
&\geq \frac{1}{10}.\mbox{\ \ \ \ (From~(\ref{deltabound}))}\label{infbound}
\end{align}
Hence, from~(\ref{infbound}) we have
\begin{align}
\Hen(g(x) \mid a_1, \ldots a_{v_{10d^2}}) \leq 1-\frac{1}{10}=\frac{9}{10}. \label{enbound}
\end{align}
Let $\cL$ be the set of leaves $\ell$ of $\clB'$ such that $\Hen(g(x) \mid \ell) \leq \frac{19}{20}$. For each $\ell \in \cL$, $\min_b \Pr_{x \sim \mu}[g(x)=b \mid x \in \ell] \leq \frac{2}{5}$. Conditioned on $(a_1, \ldots, a_{10d^2}) \in \cL$, the probability that $\clB'$ errs is at most $\frac{2}{5}$. By \emph{Markov's inequality} and (\ref{enbound}), it follows that $\Pr[(a_1, \ldots, a_{10d^2}) \in \cL] \geq \frac{1}{19}$. Thus $\clB'$ errs with probability at most $\frac{1}{19}\cdot \frac{2}{5}+\frac{18}{19}\cdot \frac{1}{2}=\frac{47}{95}$.
\end{description}
\end{proof}
\section{The Composition Theorem}
\label{comp}
In this section we prove Theorem~\ref{mainb} (restated below).
\mainb*
\begin{proof}
We shall prove that for each distribution $\eta$ on the inputs to $f$, there is a query algorithm $\cA$ making $O(\R(f \circ g^n) /\chi(g))$ queries in the worst case, for which $\Pr_{z \in \nu}[(z,\cA(z)) \in f] \geq \frac{5}{9}$ holds. This will imply the theorem by \emph{Yao's minimax principle}. To this end let us fix a distribution $\eta$ over $\{0,1\}^n$.

Let $\chi(g)=d$. Thus, there is a \emph{hard} pair of distributions $\mu_0, \mu_1$, supported on $g^{-1}(0)$ and $g^{-1}(1)$ respectively, such that for every decision tree $\clB$ that computes $g$, $\chi(\mu_0, \mu_1, g) \geq d$. We will use distributions $\eta, \mu_0$ and $\mu_1$ to set up a distribution $\gamma_\eta$ over the  input space of $f \circ g^n$. For a fixed $z=(z_1, \ldots, z_n) \in \{0,1\}^n$, We recall the distribution $\gamma_z$ over $\left(\{0,1\}^m\right)^n$ from Section~\ref{cc}. $\gamma_z$ is given by the following sampling procedure:
\begin{enumerate}
\item For $i=1, \ldots, n$, sample $x_i=(x_i^{(j)})_{j=1, \ldots, m}$ from $\mu_{z_i}$ independently for each $i$.
\item return $x=(x_i)_{i=1, \ldots, n}$.
\end{enumerate}
Now, let $\gamma_\eta$ be the distribution over $\left(\{0,1\}^m\right)^n$ that is given by the following sampling procedure:
\begin{enumerate}
\item Sample $z =(z_1, \ldots, z_n)$ from $\eta$.
\item Sample $x=(x_i)_{i=1, \ldots, n}$ from $\gamma_z$. Return $x$.
\end{enumerate}
Observe that for each $z, x$ sampled as above, for each $s \in \mathcal{S}$, $(z,s) \in f$ \emph{if and only if} $(x,s) \in f \circ g^n$.

Assume that $\R_{1/3}(f \circ g^n)=t$. Yao's mimimax principle implies that there is a deterministic query algorithm $\cA'$ for inputs from $\left(\{0,1\}^m\right)^n$, that makes at most $t$ queries in the wors case, such that $\Pr_{x \in \gamma_\nu}[(x,\cA'(x)) \in f \circ g^n] \geq \frac{2}{3}$. We will first use $\cA'$ to construct a randomized algorithm $T$ for $f$, whose accuracy is as desired, and for which the expected number of queries made is small.

\begin{algorithm}[!h]\label{T}
\DontPrintSemicolon
\caption{ $T$ on $z$}
%\KwIn{$z \in \{0,1\}^n$}
%\KwOut{$f(z)$ or $\perp$ (abort)}
\For{$1 \leq k \leq n$}
{$\nq_k \gets 1$. \;
$\N_k \gets 0$. \;}
$v \gets $Root of $\cA'$ \ \ \ \ \ \ \ \ \ \ \ \ \ \  \ \ \ \ \ \ \ // Corresponds to $\{0,1\}^m$ \;
\While{$v$ is not a leaf of $\cA'$}
{
Let $\cA'$ query $x_i^{(j)}$ at $v$. \;
\If{$\nq_i = 1$}
{Sample a fresh real number $r \sim [0,1]$ uniformly at random. \label{sampler}\;
\If{$r \leq \min_b \Pr_{x_i \sim \mu_b}[x_i^{(j)}=0 \mid x_i \in v_i]$}
{$v \gets v_0$. \;
}
\ElseIf{$r \geq \max_b \Pr_{x_i \sim \mu_b}[x_i^{(j)}=0 \mid x_i \in v^{(i)}]$}
{$v \gets v_1$. \;
 }
\Else
{$\nq_i \gets 0$. \label{query}\;
Query $z_i$. \;
\If{$r \leq \Pr_{x_i \sim \mu_{z_i}}[x_i^{(j)}=0 \mid x_i \in v^{(i)}]$}
{$v \gets v_0$. \;}
\Else
{$v \gets v_1$. \;}
}
$ \N_i \gets \N_i+1$. \;
\Else{
Sample $b$ from the distribution $\mu_{z_i}$ conditioned on the event $x_i \in v^{(i)}$. \;
$v \gets v_b$. \;}
}
}
\end{algorithm}
$T$, described formally in Algorithm~\ref{T}, is essentially viewing the process $\clP$ for $z, \mu_0, \mu_1, A'$ as a query algorithm runnng on input $z$; an assignment of $0$ to $\nq_i$ corresponds to a query to $z_i$. By Claim~\ref{samedistn}, we have that for each $z \in \{0,1\}^n$, $\Pr[(z, T(z)) \in f]=\Pr_{x \sim \gamma_z}[(x,A'(x)) \in f \circ g^n]$. Thus, $\Pr_{z \sim \eta}[(z, T(z)) \in f]=\Pr_{x \sim \gamma_\eta}[(x,A'(x)) \in f \circ g^n] \geq \frac{2}{3}$.

We now bound the expected number of queries made by $T$ on each $z$. For doing that we consider the following randomized process $Q$ that acts on $z$. Let $\clB$ be an optimal tree for distributions $\mu_0, \mu_1$. $Q$ is described formally in Algorithm~\ref{Q}.
\begin{algorithm}[!h]\label{Q}
\DontPrintSemicolon
\caption{ $Q$ on $z$}
%\KwIn{$z \in \{0,1\}^n$}
%\KwOut{$f(z)$ or $\perp$ (abort)}
Run $T$ on $z$. \;\label{runT}
\For{$1 \leq i \leq n$}
{\If{$\nq_i =1$}
{Run process $\clP$ on $\clB, \mu_0, \mu_1, x_i$ until $\nq_i$ is set to $0$. \; \label{runP}}}
\end{algorithm}
Since $\clB$ computes $g$, process Q is guaranteed to set $\nq_i$ to $0$ for each $i$. In steps~\ref{runT} and~\ref{runP}, the process $\clP$ is run with trees $A'$ and $\clB$, and the trees make queries inside the for loop of $\clP$. These queries can be thought of as being made to an $mn$ bit string $(x_i^{(j)})_{{i=1, \ldots, n}\atop{j=1, \ldots, m}}$. Let the random variable $X_i$ stand for the total number of queries made by these trees in $x_i$. $X=\sum_{i=1}^n X_i$ is the total number of queries in $Q$, i.e., the total number of iterations of the for loop of $\clP$ in all the runs of $\clP$ in $Q$. The next claim bounds $\E X$ from below.
\begin{claim}
\label{dp}
\[\E X \geq nd.\]
\end{claim}
\begin{proof}
Towards a contradiction assume that $\E X < nd$. Thus there exists an $i$ such that $\E X_i < d$. Notice that this expectation is over the random real numbers sampled in the for loop of $\clP$. Thus, there exists a fixing of those real numbers $r$ that are sampled in those iterations of the for loop of $\clP$ that correspond to queries into $x_j$ for $j \neq i$, such that conditioned on that fixing, $\E X_i < d$. However, under that fixing, process $Q$ is equivalent to process $\clP$ for some deterministic decision tree $T'$ that computes $g(x_i)$ (since $\nq_i$ is set to $0$ with probability $1$), $\mu_0, \mu_1$ and $z_i$. Thus $\E X_i < d$ conditioned on the above-mentioned fixing of randomness contradicts the assumption that $\min_\clB \chi(\clB, \mu_0, \mu_1)=\chi(g)=d$, where the minimum is taken over all deterministic decision tree $\beta$ that computes $g$. 
\end{proof}
Now, let $Y$ denote the size of the random set $\{i \mid \nq_i \mbox{ is set to $0$ in step~\ref{runT} in $Q$}\}$. Now, conditoned on the event $Y=b$, the expected number of queries made in step~\ref{runP} of $Q$ is $(n-b)d=nd-bd$. So under this conditioning the total number of queries $X$ made by $Q$ is at most $t+nd-bd$. Taking expectation over $b$, and using Claim~\ref{dp} we have that
\begin{align}
t + nd - d \cdot \E Y \geq nd \Longrightarrow \E Y \leq \frac{t}{d}. \nonumber
\end{align}
Observing that for each $z$, $Y$ has the same distribution as the number of queries made by $T$ when run on $z$, we conclude that for each $z$, $T$ makes at most $t/d$ queries on expectation. By Markov's inequality, the probability that $T$ makes more than $9t/d$ queries is at most $1/9$. Thus the probabilistic algorithm $\cA''$ obtained by terminating $T$ after $10t/d$ queries computes $f$ with probability at least $2/3 - 1/9=5/9 > 1/2$ on a random input from $\eta$. By fixing the randomness of $\cA'$ appropriately we get a deterministic algorithm $\cA$ of complexity $O(t/d)=O(\R(f \circ g)/\chi(g))$ such that $\Pr_{z \sim \eta}[(z,\cA(z)) \in f] \geq \frac{5}{9}$.
\end{proof}
\paragraph{Acknowledgements.} I thank Rahul Jain for helpful discussions.

This material is based on research supported by the Singapore National Research Foundation under NRF RF Award No. NRF-NRFF2013-13.
\bibliographystyle{plain}
\bibliography{ref}

\begin{thebibliography}{1}

\bibitem{fstcomp}
Anurag Anshu, Dmitry Gavinsky, Rahul Jain, Srijita Kundu, Troy Lee, Priyanka
  Mukhopadhyay, Miklos Santha, and Swagato Sanyal.
\newblock A composition theorem for randomized query complexity.
\newblock In {\em FSTTCS}, 2017.

\bibitem{DBLP:conf/icalp/Ben-DavidK16}
Shalev Ben{-}David and Robin Kothari.
\newblock Randomized query complexity of sabotaged and composed functions.
\newblock In {\em 43rd International Colloquium on Automata, Languages, and
  Programming, {ICALP} 2016, July 11-15, 2016, Rome, Italy}, pages 60:1--60:14,
  2016.

\bibitem{newcomp}
Dmitry Gavinsky, Troy Lee, and Miklos Santha.
\newblock On the randomised query complexity of composition.
\newblock {\em CoRR}, abs/1801.02226, 2018.

\end{thebibliography}
\end{document}